\documentclass[a4paper,11pt]{scrartcl}
\pdfoutput=1
\usepackage{a4wide}
\usepackage{diagbox}
\usepackage{amsmath,amsthm}
\usepackage{anyfontsize}
\usepackage[libertinus,lf]{newtx}
\usepackage{inconsolata}
\usepackage{bm}
\usepackage[UKenglish]{babel}
\usepackage{microtype}
\usepackage{complexity}
\usepackage{comment}
\usepackage{optidef} 
\usepackage{hyperref} 
\usepackage[nameinlink]{cleveref}
\usepackage{titling}
\usepackage{xcolor}
\usepackage{thmtools, thm-restate}

\definecolor{linkColor}{RGB}{27, 82, 140}
\hypersetup{colorlinks,allcolors=linkColor}

\renewcommand{\epsilon}{\varepsilon}
\renewcommand{\phi}{\varphi}

\DeclareMathOperator{\Exp}{\mathbb{E}}

\newcommand{\maxcut}{\textsf{\textup{Max-Cut}}}
\newcommand{\maxdicut}{\textsf{\textup{Max-DiCut}}}
\newcommand{\maxand}{\textsf{\textup{Max-And}}}
\newcommand{\maxandeven}{\textsf{\textup{Max-And-Even}}}
\newcommand{\maxdicutcut}{\textsf{\textup{Max-DiCut-Cut}}}
\newcommand{\maxdicutacyc}{\textsf{\textup{Max-DiCut-Acyclic}}}
\newcommand{\cutorient}{\textsf{\textup{Cut-Orientation}}}
\newcommand{\maxPCSP}{\textsf{\textup{Max-PCSP}}}

\newclass{\UGC}{UGC}

\declaretheorem[shaded={bgcolor=linkColor!5}]{theorem}
\declaretheorem[sibling=theorem, shaded={bgcolor=linkColor!5}]{lemma, corollary, proposition}

\theoremstyle{definition}
\declaretheorem[sibling=theorem, shaded={bgcolor=linkColor!5}]{definition}

\author{Tamio-Vesa Nakajima\\ University of Oxford\\ \texttt{tamio-vesa.nakajima@cs.ox.ac.uk} \and Stanislav \v{Z}ivn\'y\\ University of Oxford\\ \texttt{standa.zivny@cs.ox.ac.uk}}

\title{Maximum And-~vs.~Even-SAT\thanks{An extended abstract of this work will
appear in Proceedings of APPROX 2025. This work was supported by UKRI EP/X024431/1 and by a Clarendon Fund Scholarship. For the purpose of Open Access, the authors have applied a CC BY public copyright licence to any Author Accepted Manuscript version arising from this submission. All data is provided in full in the results section of this paper.}}

\date{\today}

\begin{document}

\maketitle

\begin{abstract}
A multiset of literals, called a clause, is \emph{strongly satisfied} by an
  assignment if \emph{no} literal evaluates to false. Finding an assignment that
  maximises the number of strongly satisfied clauses is \NP-hard. We present a
  simple algorithm that finds, given a multiset of clauses that admits an assignment
  that strongly satisfies  $\rho$ of the clauses, an assignment in
  which at least $\rho$ of the clauses are \emph{weakly satisfied}, in
  the sense that an \emph{even} number of literals evaluate to false.

In particular, this implies an efficient algorithm for finding an undirected cut
  of value $\rho$ in a graph $G$ given that a directed cut of value $\rho$ in $G$ is promised to exist. A similar argument also gives an efficient algorithm for finding an acyclic subgraph of $G$ with $\rho$ edges under the same promise.
\end{abstract}

\section{Introduction}\label{sec:intro}

The Maximum Cut problem in \emph{undirected} graphs (\maxcut) is a fundamental
problem, seeking a partition of the vertex set into two parts while maximising
the number of edges going across. While \maxcut{} is
\NP-hard~\cite{Karp1972}, a random assignment leads to a $\frac{1}{2}$-approximation
algorithm. In their influential work, Goemans and Williamson gave the first
improvement and presented an SDP-based $\alpha_{\text{GW}}$-approximation
algorithm~\cite{GW95}, where $\alpha_{\text{GW}}\approx 0.878$.
Under Khot's Unique Games Conjecture (\UGC)~\cite{Khot02stoc}, this
approximation factor is optimal~\cite{KKMO07,Mossel10:ann}. The current best
known inapproximability result not relying on \UGC{} is $\frac{16}{17}\approx
0.941$~\cite{Trevisan00:sicomp} (obtained by a gadget from H{\aa}stad's optimal
inapproximability result~\cite{Hastad01}).

The Maximum Cut problem in \emph{directed} graphs (\maxdicut) is a closely
related and well-studied \NP-hard problem, seeking a partition of the vertex
set into two parts while maximising the number of edges going across in the
prescribed direction. A random assignment leads to a $\frac{1}{4}$-approximation
algorithm. In the first improvement over the random assignment, Goemans and
Williamson presented an SDP-based $0.796$-approximation algorithm~\cite{GW95}.
By considering a stronger SDP formulation (with triangle inequalities), Feige
and Goemans later presented a $0.859$-approximation algorithm for
\maxdicut~\cite{FG95}, building on the work of Feige and Lov\'asz~\cite{FL92}.
Follow-up works by Matuura and Matsui~\cite{Matuura2003new} and by Lewin,
Livnat, and Zwick~\cite{Lewin02:ipco} further improved the approximation factor,
obtaining a $0.874$-approximation algorithm~\cite{Lewin02:ipco}.
On the hardness side, the best inapproximability factor under \NP-hardness is
$\frac{12}{13}\approx 0.923$~\cite{Trevisan00:sicomp} (again, via a gadget from a result
in~\cite{Hastad01}).
In recent work, Brakensiek, Huang, Potechin, and Zwick gave an
$0.8746$-approximation algorithm for \maxdicut, also showing that it is \UGC{}-hard
to achieve a $0.875$-approximation~\cite{Brakensiek23:focs}.

\paragraph{Our results.}
Consider the following \emph{promise} variant of the two discussed problems:

\begin{definition}[$\maxdicutcut$]
Given a
directed graph that has a directed cut of value $\rho$, efficiently find
an undirected cut (i.e., ignoring edge directions) of value at
least $\rho$.
\end{definition}

It turns out that the $\maxdicutcut$ admits an efficient algorithm, as will follow from our
more general result (cf.~\Cref{thm:main} below).

We represent the Boolean \emph{true} value by $+1$ and the \emph{false} value by $-1$. 
A \emph{literal} $sx$ is a variable $x$ ($s=1$, positive sign) or its
negation $-x$ ($s=-1$, negative sign).
A \emph{clause} $C=\{s_1x_1,\ldots,s_k x_k\}$ is a multiset of literals. An assignment
of $+1$s and $-1$s to the variables of a clause $C$ \emph{strongly
satisfies} $C$ if no literal evaluates to false, and \emph{weakly satisfies} $C$
if an even number of literals evaluates to false.

We now define a natural variant of the Boolean satisfiability problem.

\begin{definition}[$\maxandeven$]
Given a
multiset of clauses for which there is an assignment strongly satisfying 
$\rho$ of the clauses, find an assignment that weakly satisfies 
$\rho$ of the clauses.
\end{definition}

Coming back to $\maxdicutcut$, a directed edge $(u,v)$ is
cut if the clause $\{-u,v\}$ is strongly satisfied, and the edge is cut (ignoring
the direction) if the clause $\{-u,v\}$ is weakly satisfied. Thus, $\maxandeven$ is a generalisation of $\maxdicutcut$.

\medskip
Our main result is an algorithm for $\maxandeven$, and thus also for $\maxdicutcut$.

\begin{restatable}{theorem}{thmMain}\label{thm:main}
There exists a polynomial-time algorithm for $\maxandeven$.
\end{restatable}

\noindent
\Cref{thm:main} has an interesting corollary for a related, and in some sense dual, problem.

\begin{definition}[\cutorient{}]
Given an undirected graph $G$ that has a
maximum cut of value $\rho$, orient the edges of $G$ so that the
resulting directed graph has a directed cut of maximum possible value.
\end{definition}

\cutorient{} can be approximated with the ratio $\alpha_{\text{GW}}$: Find a cut in
$G$ of size $\alpha_{\text{GW}} \rho$ using the Goemans-Williamson
algorithm~\cite{GW95}, then orient all the edges from one side of the cut to the other. 
Interestingly, \Cref{thm:main} shows that this approximation is \UGC{}-optimal.

\begin{corollary}
It is \UGC{}-hard to approximate \cutorient{} with a factor better than $\alpha_{\text{GW}}$.
\end{corollary}

\begin{proof}
The observation is that the $\maxdicutcut{}$ problem gives a reduction from \maxcut{} to \cutorient:
Given an instance of \maxcut{} that has a cut of size $\rho$, if we could orient
the graph to obtain a directed cut of size $\alpha\rho$, then by solving the resulting digraph as an instance of \maxdicutcut{} we would find a cut of size $\alpha\rho$. 
Since it is \UGC-hard to approximate $\maxcut$ with a factor better than $\alpha_{\text{GW}}$~\cite{KKMO07,Mossel10:ann},
the same is true for \cutorient{}.
\end{proof}

\noindent
Using ideas from the proof of~\Cref{thm:main}, we give an efficient algorithm for another problem.

\begin{definition}[$\maxdicutacyc$]
Given a directed graph $G$ that has a directed cut of value $\rho$, 
efficiently find an acyclic subgraph of $G$ with at least $\rho$ edges.
\end{definition}

\begin{restatable}{theorem}{thmMainTwo}\label{thm:main2}
There exists a polynomial-time algorithm for $\maxdicutacyc$.
\end{restatable}

\paragraph{Related work.}
The motivation for our work comes from a systematic study of so-called promise constraint satisfaction problems (PCSPs)~\cite{BG21:sicomp,BBKO21}. These are problems concerned with homomorphisms between graphs and, more generally, relational structures. A homomorphism $h$ from a graph $G$ to a graph $H$, also known as an $H$-colouring of $G$~\cite{Hell90:h-coloring}, is a map from the vertex set of $G$ to the vertex set of $H$ that preserves edges; i.e., if $(u,v)$ is an edge in $G$ then $(h(u),h(v))$ must be an edge in $H$. For instance, if $H=K_3$ is a clique on 3 vertices, then homomorphisms from $G$ to $H$ are precisely 3-colourings of $G$. Equivalently, there is a homomorphism from $G$ to $H$ if $G$ is a subgraph of a blow-up of $H$, where a blow-up of $H$ replaces every vertex by an independent set and every edge by a complete bipartite graph. Homomorphisms between relational structures are defined analogously as maps from the universe of one structure to the universe of another structure so that all relations are preserved by a component-wise application of the map. 

Our work is related to an optimisation variant of PCSPs.
In particular, the problem $\maxPCSP(G,H)$ asks: Given an input structure $X$ which admits a $G$-colouring of value $\rho$, find an $H$-colouring of value at least $\rho$. For example, $\maxPCSP(G,H)$ with bipartite $G$ was classified~\cite{nz25:icalp} and the intractability of non-bipartite $G$ was very recently established in~\cite{nz25:arxiv-dichotomy}. This leaves open cases where $G$ and $H$ are not graphs but rather arbitrary relational structures. The simplest open case is Boolean binary structures, i.e., digraphs on 2 vertices. There are three interesting problems: $\maxcut$, $\maxdicut$, and $\maxdicutcut$. Since $\maxcut{}$ and $\maxdicut{}$ are already well-understood, understanding the complexity of 
$\maxdicutcut$ was the first step in the ultimate goal of understanding all
structures beyond (undirected) graphs. Thus, we view the importance
of~\Cref{thm:main} as twofold. Firstly, as a non-trivial tractability result for
a natural problem. Secondly, as initiating the study of $\maxPCSP$ beyond
graphs. Finally, \Cref{thm:main2} gives a tractability result for
$\maxPCSP({A},{B})$, 
where ${A}=(\{0,1\},\{(0,1)\})$ represents the cut problem in directed graphs and ${B}=(\mathbb{Q},<)$ represents the graph acyclicity problem, thus identifying a natural tractable $\maxPCSP$ with an infinite structure, following a well-established line of work on infinite-domain CSPs~\cite{Bodirsky2021:book}.

\section{\texorpdfstring{Algorithm from~\Cref{thm:main}}{Algorithm from Theorem 3}}\label{sec:alg1}

We denote by $[k]$ the set $\{1,\ldots,k\}$.
For a statement $\phi$, we let $[\phi]$ be 1 if $\phi$ is true and $0$ otherwise. 
For a clause $C = \{ s_1 x_1, \ldots, s_k x_k \}$, weak and strong satisfaction
of $C$ by an assignment $c$ can be expressed as
\begin{align*}
    [ \text{$C$ is strongly satisfied} ] &= \frac{1}{2} + \frac{1}{2} \min_{1 \leq i \leq k} s_i c(x_i), \\
    [ \text{$C$ is weakly satisfied} ] &= \frac{1}{2} + \frac{1}{2} \prod_{1 \leq i \leq k} s_i c(x_i).
\end{align*}

Thus, a clause $C$ is strongly satisfied if and only if the minimum of
all the literals in that clause is $1$, and this minimum is $-1$ otherwise.
Hence, we can cast \maxand{} over an instance with variable set $V$, and clauses
$C_1, \ldots, C_m$ with $C_i = \{ s_1^i x_1^i, \ldots, s^i_{k_i} x^i_{k_i} \}$, as the problem of finding a solution to
\begin{equation}\label{eq:thing}
\max_{c : V \to D} \sum_{i = 1}^m \frac{1}{2} + \frac{1}{2} \min_{1 \leq j \leq
  k_i} s^i_j c(x^i_j),
\end{equation}
where we take $D = \{-1,+1\}$. One way to establish~\Cref{thm:main} would be to
first solve this problem relaxed to $D = [-1, 1]$ using LP, and then round directly. 
We take a slightly different (and simpler) route, which relies on the idea of ``half-integrality'' and will be also useful for
proving~\Cref{thm:main2}:~\eqref{eq:thing} can be solved exactly if we take $D = \{-1, 0, +1\}$.

\begin{theorem}\label{thm:solver}
    An optimal solution to~\eqref{eq:thing} can be found in polynomial time for $D = \{-1, 0, +1\}$.
\end{theorem}
\begin{proof}
    First, find any optimal solution $c^*$ to~\eqref{eq:thing} with $D = [-1, 1]$ by LP.\@ To do this, one can use a standard trick. For each clause $C_i = \{ s_1^i x_1^i, \ldots, s_{k_i}^i x_{k_i}^i \}$ for $i \in [m]$, introduce a variable $t_i$ and constraints $t_i \leq s_j^i c(x_j^i)$ for $j \in [k_i]$. Next, replace the objective function by $\sum_{i = 1}^m \frac{1}{2} + \frac{1}{2} t_i$. Observe that in any feasible solution to this LP, we have $t_i \leq \min_{1 \leq j \leq k_i} s_j^i c(x_j^i)$. Furthermore, since the objective is an increasing function of $t_1, \ldots, t_m$, in any optimal solution to this LP we have $t_i = \min_{1 \leq j \leq k_i} s_j^i c(x_j^i)$. Hence the optimal solutions to this LP precisely correspond to the optimal solutions to~\eqref{eq:thing}.
    
    We will shift the solution while keeping it optimal until the image of $c^*$ is contained in $\{-1, 0, 1\}$. As a pre-processing step, flip signs of literals so that $c^*(x) \geq 0$ for $x \in V$. Our goal now is to have the image of $c^*$ contained in $\{0, 1\}$. Suppose without loss of generality that $V = \{ 1, \ldots, n \}$ and $c^*(1) \leq \cdots \leq c^*(n)$. Define
    \[
    F(c) = \sum_{i = 1}^m \frac{1}{2} + \frac{1}{2} \min_{1 \leq j \leq k_i} s_j^i c(x_j^i).
    \]

    Note that for every $c$ for which $0 \leq c(1) \leq \cdots \leq c(n)$, $\arg \min_{1 \leq j \leq k_i} s_j^i c(x_j^i)$ is known. 
    Indeed, if all $s_j^i$ are positive then the minimum is attained at the
    smallest $x_j^i$; whereas if there is at least one negative $s_j^i$ then the
    minimum is attained at the largest $x_j^i$ among those with $s_j^i=-1$.
    (Here we compare variables by the natural ordering, since we have assumed the variables are natural numbers.) Let $j(i) = \arg \min_{1 \leq j \leq k_i} s_j^i c(x_j^i)$ be this minimum; then, for all $c$ where $0 \leq c(1) \leq \cdots \leq c(n)$,
    \[
    F(c) = \sum_{i = 1}^m \frac{1}{2} + \frac{1}{2} s_{j(i)}^i c(x_{j(i)}^i),
    \]
    In other words, $F$ is an affine function in terms of $c$! Define $c^*(0) = 0, c^*(n + 1) = 1$. While $c^*$ has an image that is not 0 or 1, do the following. Take $1 \leq a \leq b \leq n$ so that $0 = c^*(0) = \cdots = c^*(a -1) < c^*(a) = \cdots = c^*(b) < c^*(b + 1) \leq \cdots \leq c^*(n + 1) = 1$. Such $a, b$ exist, or else $c^*(x) \in \{0, 1\}$ for all $x \in [n]$.
    If the sum of the coefficients of $c^*(a), \ldots, c^*(b)$ in $F$ (seen as an affine function) is negative, then we could decrease all of these values together (since this maintains the fact that $0 \leq c^*(1) \leq \cdots \leq c^*(n)$) and improve our solution, which is impossible. If the sum of the coefficients is positive, then we could improve our solution by increasing the values of $c^*(a), \ldots, c^*(b)$, which is impossible. So we can assume the sum of the coefficients is 0, in which case we can set the values $c^*(a), \ldots, c^*(b)$ to anything we want in the interval $[0, c^*(b + 1)]$ without changing the value of the solution. So set $c^*(a) = \cdots = c^*(b) = 0$. Continue this procedure until there are no variables set to anything other than  0 or 1.
\end{proof}

We now prove~\Cref{thm:main}, restated below for the reader's convenience. 

\thmMain*

\begin{proof}
Without loss of generality, we can assume that each variable appears in each clause at most once. Indeed, if a clause $C$ contains a variable $x$ and its negation $-x$ then $C$ cannot be strongly satisfied but could be weakly satisfied, so running our algorithm on the instance without $C$ causes no issues. Similarly, if a literal appears twice in a clause then both occurrences can be removed, as this does not affect weak satisfiability (but could improve strong satisfiability). 

Suppose we are given an instance of \maxandeven{}, namely an expression of form
\[
    \max_{c : V \to \{ \pm 1 \}} \sum_{i = 1}^m \frac{1}{2} + \frac{1}{2} \min_{1 \leq j \leq k_i} s^i_j c(x^i_j).
\]
Suppose that the value of this expression is $\rho$. Our goal is to find a
  function $c^*$ such that that number of weakly satisfied clauses is at least $\rho$. Recalling that a clause $\{ s_1 x_1, \ldots, s_k x_k \}$ is weakly satisfied by $c$ if and only if $\frac{1}{2} + \frac{1}{2} \prod_i s_i c(x_i) = 1$, and that this expression is 0 otherwise, we note that we must find $c^* : V \to \{-1, +1\}$ such that
\begin{equation}\label{eq:goal}
    \rho \leq \sum_{i = 1}^m \frac{1}{2} + \frac{1}{2} \prod_{1 \leq j \leq k_i} s^i_j c^*(x^i_j).
\end{equation}

\noindent
Using~\Cref{thm:solver}, we can efficiently find a function $\tilde{c} : V \to \{-1, 0, +1\}$ such that

\[
    \rho \leq \sum_{i = 1}^m \frac{1}{2} + \frac{1}{2} \min_{1 \leq j \leq k_i} s^i_j \tilde{c}(x^i_j).
\]
We claim that the following choice of $c^*$ makes~\eqref{eq:goal} hold in
  expectation: if $\tilde{c}(v) = \pm 1$ then set $c^*(v) = \tilde{c}(v)$.
  Otherwise set $c^*(v)$ equal to $+1$ or $-1$ uniformly and independently at
  random. Indeed, by linearity of expectation it suffices to show that, for any clause $\{ s_1 x_1, \ldots, s_k x_k \}$, we have
\[
\frac{1}{2} + \frac{1}{2}\min_{1 \leq i \leq k} s_i \tilde{c}(x_i)
\leq
\Exp\left[
\frac{1}{2} +\frac{1}{2} \prod_{1 \leq i \leq k} s_i c^*(x_i).
\right]
\]
In particular, let us consider the value of the minimum on the left. If it is
  $-1$ there is nothing left to prove. If it is $+1$, then $\tilde{c}(x_i) \neq
  0$ and hence $c^*(x_i) = \tilde{c}(x_i)$ for all $x_i$, and the bound holds
  with equality. Finally, if the minimum is 0, then there exists some $x_i$ such
  that $\tilde{c}(x_i) = 0$; hence $c^*(x_i)$ is set to $+1$ or $-1$ equiprobably. It is easy then to see that the product on the right hand side is either $+1$ or $-1$ equiprobably, whence the conclusion. (This depends, essentially, on the fact that every variable appears in each clause at most once.)

We now derandomise the rounding scheme above; in other words, we will show how
to deterministically set the variables so that the resulting
  value is at least as good as the expectation of the random variables. 
Equivalently, given a set of parity constraints on
subsets of variables, the goal is to find an assignment that satisfies at least as
many constraints as the random assignment. This problem can be derandomised easily 
using the method of conditional expectation --- indeed these techniques have appeared before in e.g.~\cite{KSTW00}, but we include them for completeness.
Recall that, by assumption,
each variable appears in each clause at most once. Thus, conditional on some
subset of variables being fixed, we expect half of the remaining parity constraints
that have at least one unfixed variable within them to be satisfied.
Thus, we can derandomise by going through the variables one by one; when we
set a variable $x$, we set it so that as many as possible of the parity constraints where $x$
is the last unfixed variable remaining are satisfied.
\end{proof}

\section{\texorpdfstring{Algorithm from~\Cref{thm:main2}}{Algorithm from Theorem 7}}

\noindent
We are now ready to prove~\Cref{thm:main2}, restated below.

\thmMainTwo*

\begin{proof}
    Equivalently, given a digraph $G$ that has a dicut with $\rho$ edges, we will produce an ordering of the vertices of $G$ such that at least $\rho$ edges are oriented according to the ordering --- removing all the edges that are not oriented in this way gives us the required subgraph. For some ordering $\sigma$, we say that an edge $(x, y)$ is well ordered by $\sigma$ if $x$ comes before $y$ in $\sigma$.

    As a pre-processing step, remove all loops from the input digraph. This will never lower the size of the maximum dicut, nor will it lower the size of the outputted subgraph. Now compute a solution $\tilde{c}$ to~\eqref{eq:thing} for $G = (V, E)$ over $D = \{-1, 0, 1\}$, using \Cref{thm:solver}. By this, we mean a solution to the instance including the clause $\{-x, y\}$ once for each occurrence of the edge $(x, y)$ in $G$.
    In other words, each edge $(x, y)$ contributes $\frac{1}{2} + \frac{1}{2} \min(-c(x), c(y))$ to~\eqref{eq:thing}. The values of this expression over $D = \{-1, 0, 1\}$ are tabulated in~\Cref{fig:expression}.

\begin{figure}[!h]
\centering
\begin{tabular}{c|ccc}
     \diagbox{${c}(x)$}{${c}(y)$} & $-1$ & $0$ & $1$ \\ \hline
     $-1$ & $0$ & $\frac{1}{2}$ & $1$ \\
     $0$ & $0$ & $\frac{1}{2}$ & $\frac{1}{2}$ \\
     $1$ & $0$ & $0$ & $0$
\end{tabular}
\caption{Value of $\frac{1}{2} + \frac{1}{2} \min(-c(x), c(y))$}\label{fig:expression}
\end{figure}
  
    Since the digraph admits a dicut with $\rho$ edges, the value of $\tilde{c}$ is at least $\rho$. Partition the vertices of $G$ into $V_{-1}, V_0, V_1$ according to their image through $\tilde{c}$. Let $E_{ij} = E \cap (V_i \times V_j)$, for $i, j \in \{-1, 0, 1\}$. Let $\pi_{-1}, \pi_0, \pi_1$ be arbitrary orderings of $V_{-1}, V_0, V_1$, and let $\pi_0'$ be the reverse of $\pi_0$. We claim that one of the orderings $\sigma = (\pi_{-1}, \pi_0, \pi_1)$ or $\sigma' = (\pi_{-1}, \pi_0', \pi_1)$ can be returned.

    To see why, note that every edge $(x, y)$ with $\tilde{c}(x) < \tilde{c}(y)$ is well ordered by both $\sigma$ and $\sigma'$. Furthermore, at least half the edges $(x, y)$ with $\tilde{c}(x) = \tilde{c}(y) = 0$ are ordered correctly in one of $\sigma$ or $\sigma'$ (this is why removing loops is essential). But (cf.~\Cref{fig:expression} and the optimality of $\tilde{c}$),
    \begin{multline*}
    \rho
    \leq \#E_{-1, 1}
    +\frac{1}{2} \# E_{-1, 0}
    +\frac{1}{2} \#E_{0, 1}
    +\frac{1}{2} \#E_{0, 0}
    \leq \#E_{-1, 1}
    + \#E_{-1, 0}
    + \#E_{0, 1}
    +\frac{1}{2} \#E_{0, 0} \\
    = \# \{ (x, y) \in E \mid \tilde{c}(x) < \tilde{c}(y) \}
    +
    \frac{1}{2} \# E_{0, 0}.
    \end{multline*}
    Hence, at least one of $\sigma$ and $\sigma'$ well orders at least $\rho$ edges.
\end{proof}

\section{Conclusions}

As discussed briefly in~\Cref{sec:intro}, the main contribution of this paper is
twofold. Firstly, we give a simple, efficient algorithm for a very natural
computational problem. Secondly, we initiate the study of 
$\maxPCSP$s beyond graphs (which have been recently classified~\cite{nz25:arxiv-dichotomy}) and beyond finite structures. A
concrete question worthy of investigating is for which $\maxPCSP$s our method of
rounding, relying on the idea of half-integrality, is applicable.

\paragraph{Acknowledgements.}
We thank anonymous reviewers of an earlier version of this paper for useful
comments and a simplification of our algorithm. We also thank the anonymous
reviewers of APPROX 2025 for their feedback and suggestions for changes.

{%
\bibliography{nz}
\bibliographystyle{alphaurl}
}

\end{document}